\documentclass[11pt]{article}
\usepackage[T1]{fontenc}
\usepackage[utf8]{inputenc}
\usepackage{lmodern,microtype,amsmath,amssymb,amsthm,booktabs,tabularx,graphicx}
\usepackage[margin=1in]{geometry}
\usepackage[round,authoryear]{natbib}
\usepackage[colorlinks=true,linkcolor=black,citecolor=black,urlcolor=blue]{hyperref}
\usepackage{enumitem,listings}
\hypersetup{pdftitle={When Old Facts Return: Re-Reads, Reverts, and the Limits of Temporal Memory},pdfauthor={Neeraj Yadav}}
\setlist{itemsep=3pt,topsep=5pt}
\newtheorem{proposition}{Proposition}
\title{When Old Facts Return:\\Re-Reads, Reverts, and the Limits of Temporal Memory}
\author{Neeraj Yadav\\Called It Inc.\\\url{https://memstrata.dev}}
\date{October 2026}
\begin{document}
\maketitle
\begin{abstract}
A memory system can retire an obsolete value and later restore it merely because the same old statement appears again. A re-read of an old source and a genuine revert can produce the same observed sequence of values while requiring opposite current answers. We study this ambiguity on 130 extractor-selected atomic transitions derived from software fixes. In the ordinary transition condition, identity-based temporal memory reaches 98.5\% model-judged accuracy with zero observed errors under a literal stale-value proxy. Appending a verbatim re-read of the old statement reduces accuracy to 10.8\% and raises the stale-value rate to 88.5\%. A guard that refuses to reactivate a previously retired value restores accuracy to 97.7\% and reduces that rate to 0.8\% in this constructed re-read condition. The guard cannot also recognize a legitimate revert without additional change provenance. Two supporting studies examine exposing retired history to the answer model and supplying current source for changed behavior. An exploratory extraction study over 707 software fixes provides scope context, not a universal coverage estimate. The design implication is to distinguish an observation of a value from evidence that the value changed. Selected inputs, aggregate-only answer records, related-family judges and a post-failure guard evaluation limit the conclusions to the retained experiments.
\end{abstract}

\section{Introduction}
A useful memory must do more than remember that a statement occurred. It must help a downstream system determine whether the statement still applies. Retrieval-augmented generation provides a way to locate relevant material \citep{lewis2020retrieval}, but a relevant passage may describe an obsolete state. Temporal memory addresses that problem by retaining identity and the relationship between successive values.

The simple transition case is straightforward. If a configuration changes from value $A$ to value $B$, the memory can retain $B$ as current and preserve $A$ in history. Earlier temporal-memory studies evaluated this operation on evolving knowledge and a selected atomic subset of software histories \citep{yadav2026temporal,yadav2026software}. The next observation makes the problem more interesting: an agent opens an old document containing $A$. Should that observation restore $A$?

Sometimes it should. A verified revert can legitimately return a project to an earlier state. Sometimes it should not. Re-reading an obsolete file adds an observation of history without changing the project. Values and mention timestamps alone cannot distinguish these cases. A memory that always prefers the latest mention can resurrect stale information. A memory that never reactivates an old value can suppress a real change.

The central contribution is the measured re-read failure and the provenance requirement exposed by the corresponding guard. The taxonomy and extraction experiments help delimit the class of software changes to which atomic supersession applies; they do not support a claim that a fixed fraction of all software evolution has been solved. We reuse the 130-scenario cache and ordinary-transition setting of the earlier software study; neither is a new independent replication here. The additional recency conditions, re-read perturbation and guard, consumer-context studies, and exploratory declined-set analysis form the present investigation.

We ask three questions. First, do recency baselines preserve both relevance and currency on accepted atomic changes? Second, what happens when an old value is observed again after its replacement? Third, how should a memory expose history and handle changes that are not atomic values? The experiments are descriptive single-run studies under retained local-model protocols. They support a narrow architectural argument, not a general product ranking.

\section{Background and scope}
\subsection{Temporal records and belief revision}
Bitemporal data models distinguish when a value applies from when a system records it \citep{snodgrass1995tsql2}. This distinction is directly relevant to repeated observations: the time of ingestion is not necessarily the time of a world change. Truth-maintenance and belief-revision systems likewise distinguish stored beliefs, their justification and the consequences of revision \citep{doyle1979tms,agm1985}. Event-based formulations provide another way to represent change explicitly \citep{kowalski1986events}.

We do not claim to introduce these distinctions. The contribution is an operational study of their consequences for a retrieval memory consumed by a language model, including a failure of the unguarded temporal implementation. The system needs more than a current-value table if its inputs do not distinguish change from observation.

\subsection{Software changes are heterogeneous}
Some changes replace a literal, endpoint or declared parameter. Others alter control flow, interactions among functions or an algorithm. Fine-grained source differencing can identify syntactic changes \citep{fluri2007changedistiller,falleri2014gumtree}; contract-oriented representations can describe selected interface changes \citep{meyer1992contract}. Neither operation makes every behavioral change reducible to a single value-bearing triple.

Our source population comes from SWE-bench Lite and Verified \citep{jimenez2024swebench}. It consists of issue-linked fixes in a limited set of mature Python repositories. It is not a random sample of all code evolution, and the accepted atomic set was selected by an earlier extraction pipeline. That selection is central to the interpretation of the results.

\section{An information boundary for current-state memory}
\subsection{Observation model}
Let $k$ denote a normalized subject-relation key and $v$ a value. An observation is $o=(k,v,t_o)$, where $t_o$ is observation time. A change-aware input may additionally provide an effective time $t_v$ and provenance $p$ identifying an authoritative change event. A current-state resolver returns a value, a conflict or an unknown result.

For ordinary monotonic changes, the transition $A\rightarrow B$ closes the earlier active interval and inserts the new value atomically. Duplicate observations of the same active value need not create a new state change. The difficult sequence is $A,B,A$ when the final $A$ lacks a trusted event classification.

\begin{proposition}[Re-read and revert indistinguishability]
Suppose a resolver observes only key, value and mention order. There exist a re-read history and a revert history with identical observations for which the correct current values differ. No deterministic resolver using only those observations can return the correct current value in both histories.
\end{proposition}
\begin{proof}
Construct two histories with observations $(k,A,t_1)$, $(k,B,t_2)$ and $(k,A,t_3)$, with $t_1<t_2<t_3$. In the first history, the state changes to $B$ at $t_2$ and the last observation re-reads an old source. The correct current value is $B$. In the second, an authoritative change at $t_3$ restores $A$, so the correct current value is $A$. The resolver receives the same input sequence in both cases and therefore returns the same output. That output cannot equal both distinct required values.
\end{proof}

The proposition states an input limitation, not a new result in formal epistemology. Its engineering consequence is concrete: a resolver must obtain additional provenance, accept a policy assumption, or abstain when the distinction matters. Replacing a similarity threshold with a more confident model does not supply information absent from the input.

\subsection{The evaluated guard}
The no-resurrect guard checks whether an incoming value for a key already appears among that key's retired values. If so, it does not make the value active merely because it was observed again. On histories where an old value never legitimately returns, this policy handles the re-read case. On a real revert, it requires an additional trusted change signal or it keeps the wrong current value.

The guard is evaluated as an opt-in mechanism. It is not a universal policy recommendation. We deliberately report its correctness cost with its re-read benefit rather than relegating the revert ambiguity to an unspecified future improvement.

\subsection{Current and historical views}
A current-value context should prefer the evidence establishing the active state. A historical query intentionally exposes retired values and their times. Presenting both views in one prompt can reintroduce the stale value as an answer candidate even when the storage layer resolved the state correctly. We evaluate that consumer boundary separately.

\section{Data and methods}
\subsection{Populations}
The union of the retained SWE-bench Lite and Verified files contains 707 distinct issue IDs. The earlier pipeline accepted 130 clean atomic transitions and declined 577. The 130-item set supplies the ordinary and re-read experiments. The declined set supplies exploratory taxonomy and contract-extraction analyses.

Each accepted scenario has an old value $A$, a new value $B$ and a question whose current answer is $B$. These are cached model-extracted, rendered statements, not the original issue-and-patch histories. Text is marker-free: the input statements are not helpfully labeled ``stale'' and ``current.'' The ordinary run ingests all 130 old statements followed by all 130 new statements into a shared store. The re-read run appends all 130 old statements verbatim, producing 390 turns and per-key order $A,B,A$, while the current target remains $B$. This shared-store construction allows cross-key retrieval interference. We do not claim full independence among questions.

\subsection{Compared conditions}
Naive RAG retrieves from retained statements without temporal retirement. Advanced RAG adds the evaluated language-model reranker. Global-recency RAG reorders a retrieved pool by ingestion time. The recency reranker receives an ordered candidate list and instructions favoring recent relevant statements. Temporal memory resolves identity-based supersession, with the no-resurrect guard disabled or enabled for the relevant condition.

The recency baselines represent specific tested implementations. A global-recency sort can sacrifice relevance, while a model instructed about recency can fail to select the right value. Their results do not establish that every recency-aware retriever must fail in the same way.

\subsection{Models, scoring and reporting}
The answer and write/read verification model is \texttt{qwen2.5-coder:7b}. Correctness judgments use \texttt{qwen2.5-coder:3b} \citep{hui2024qwen}. Embeddings use \texttt{nomic-embed-text} with 768 dimensions \citep{nussbaum2024nomic}. Models are served locally through Ollama. The retained defaults specify temperature zero and seed zero. The answer and correctness models differ in size but belong to the same model family; this is not independent-family validation. The forced-answer instruction requests an answer but does not guarantee that the model attempts every item. Appendix~\ref{app:protocol} gives the prompt and scoring rules.

Accuracy is the fraction judged correct. The reported \emph{stale-answer error} is a literal-value proxy: after lowercasing the answer, the old value must occur as a substring and the current value must not. It is not a semantic judgment of whether the answer asserts the old value as current. It can miss paraphrases and answers containing both values, and can flag an old value mentioned in a negation. Attempted answers use the retained abstention-string heuristic. These quantities are different: a system can lower stale error by refusing more often while becoming less useful.

The results are retained single-run measurements. Exact model-blob digests, complete request receipts and an execution-environment image were not preserved for all studies. The no-resurrect guard was evaluated after the unguarded failure was observed, on the same cases, rather than on a held-out confirmatory set. We therefore avoid claims of bitwise environmental replay, preregistered inference or repeated-run significance. Aggregate logs and the accompanying offline verifier reproduce the reported descriptive quantities; they cannot independently rejudge unavailable individual answers.

\section{Recency in the ordinary transition condition}
Table~\ref{tab:plain} shows that the tested recency methods do not preserve both relevance and currency on the selected set. Temporal memory reaches 98.5\% model-judged accuracy and zero detected stale-value matches. The global-recency sort achieves low stale error but attempts only 72 of 130 questions.

\begin{table}[ht]
\centering
\caption{Ordinary marker-free transitions, $n=130$. Values are retained rounded rates. The requested forced-answer regime still permits observed model abstentions.}
\label{tab:plain}
\begin{tabular}{lrrr}\toprule
Condition & Accuracy & Stale error & Attempted\\\midrule
Naive RAG & .615 & .361 & 130\\
Advanced RAG & .592 & .377 & 129\\
Global-recency RAG & .508 & .008 & 72\\
Recency reranker & .592 & .415 & 129\\
Temporal memory & .985 & .000 & 129\\\bottomrule
\end{tabular}
\end{table}

The global sort promotes recent statements that may concern another key. For a question about a less recently mentioned subject, the current relevant value can disappear from the selected context. This is a mechanism-level explanation consistent with the shared-store design and abstention counts, not a retained per-question causal audit. The low stale rate does not represent a free improvement. The language-model recency reranker also does not eliminate stale-value matches in this protocol.

These conditions motivate identity-aware state handling, but they do not yet test whether the latest observation is a real change. The next experiment removes that assumption.

\section{When an old statement returns}
Appending an old statement after its replacement changes the outcome sharply. The unguarded temporal implementation treats the final old value as a new contradiction and reactivates it. It becomes worse than naive RAG on stale-answer error in this condition.

\begin{table}[ht]
\centering
\caption{The re-read condition, $n=130$, with observation order $A,B,A$ and current target $B$. RAG baseline results are retained unchanged when the assertion guard is enabled.}
\label{tab:reread}
\begin{tabular}{lrr}\toprule
Condition & Accuracy & Stale error\\\midrule
Naive RAG & .531 & .469\\
Global-recency RAG & .069 & .739\\
Temporal memory, guard disabled & .108 & .885\\
Temporal memory, guard enabled & .977 & .008\\\bottomrule
\end{tabular}
\end{table}

The guard reduces stale-answer error by 87.7 percentage points in this constructed re-read condition. Its benefit comes from the policy assumption that a previously retired value should not regain authority merely by reappearing. It does not come from better semantic similarity or stronger downstream answer reasoning.

\begin{figure}[ht]
\centering\includegraphics[width=\linewidth]{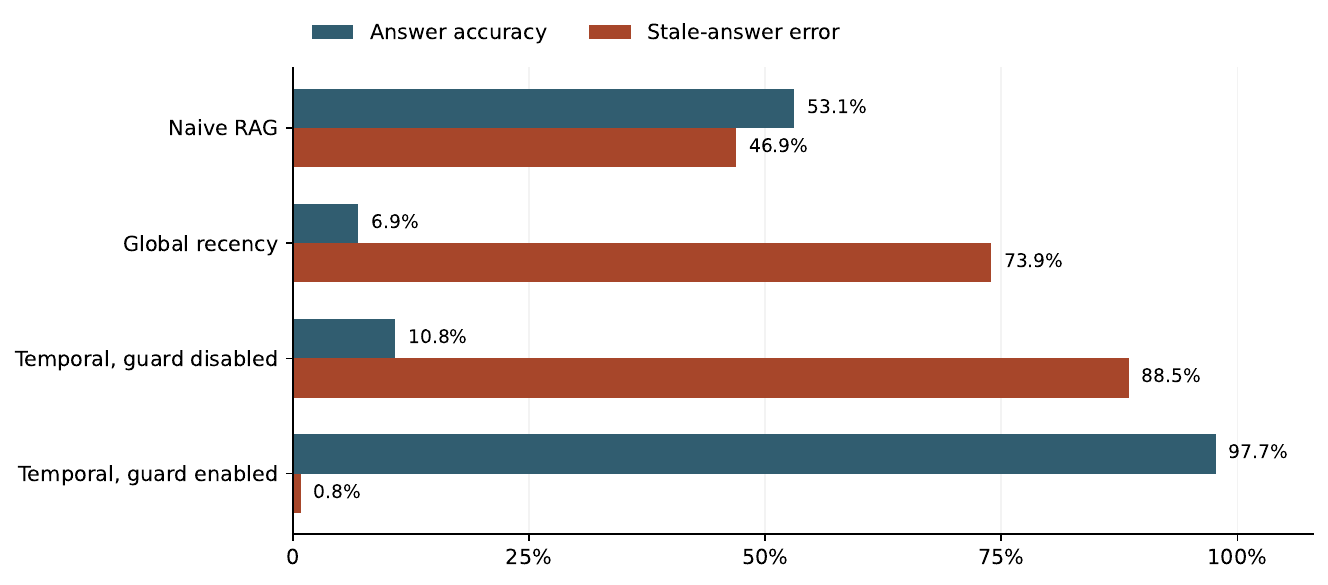}
\caption{Re-reading an old value defeats the unguarded ledger. The guard succeeds under the re-read assumption, while a genuine revert remains a provenance-dependent case.}
\end{figure}

The genuine-revert counterexample is not an additional measured result in this table. It follows from the guard's stated behavior and the indistinguishability construction. A complete deployment must decide which sources can establish a revert, how their authority is verified and what to return when the authority is unavailable.

\section{Two further boundaries at the consumer}
\subsection{Historical explanations can damage current answers}
A separate retained experiment presents the same 130 current-value questions with either active-only context or a context saying the value ``was $A$; is now $B$.'' Active-only context reaches .954 accuracy and .008 stale error. Including retired history reaches .846 accuracy and .146 stale error.

The result supports separate current and historical read modes. A correct temporal store does not ensure that a language model will disregard an obsolete value if the value is reintroduced into its answer context. The experiment does not imply that historical explanations are always harmful; their purpose and the requested answer differ in a genuine historical query.

\subsection{Behavioral changes require current source}
The logic-staleness mechanism test contains 40 synthetic scenarios. Serving only the stale summary reaches .000 accuracy and .475 stale error. Serving current raw source reaches .975 accuracy and .000 stale error. Serving current source with the evaluated warning reaches 1.000 accuracy and .000 stale error.

Most of this improvement comes from supplying current source. The one-item difference between raw source and raw source plus warning is not evidence of a general warning advantage. The useful mechanism is deterministic invalidation of a derived summary followed by retrieval of current evidence.

This is a synthetic context intervention, not an end-to-end test of a change detector. Twenty cases change a threshold literal and twenty change an arithmetic operator; all are generated from a small set of templates. The warning condition also retains the stale summary, whereas the raw-source condition does not, so their contrast changes two context components. The study does not establish real-project behavioral coverage or show that arbitrary logic has been converted into atomic facts. Dependency propagation similarly needs its own coverage analysis, particularly for dynamic edges that static source analysis cannot recover.

\begin{table}[ht]
\centering
\caption{Distinct consumer studies. The 130-item historical-context study and 40-item synthetic logic study are not pooled.}
\begin{tabular}{llrr}\toprule
Study & Context & Accuracy & Stale error\\\midrule
Current questions, 130 & Active only & .954 & .008\\
Current questions, 130 & Active plus history & .846 & .146\\
Synthetic logic, 40 & Stale summary only & .000 & .475\\
Synthetic logic, 40 & Current raw source & .975 & .000\\
Synthetic logic, 40 & Current source and warning & 1.000 & .000\\\bottomrule
\end{tabular}
\end{table}

\section{Exploratory scope of atomic extraction}
The earlier pipeline accepted 130 of 707 fixes, or 18.4\%, as clean atomic transitions. A later full-file contract extractor processed 576 of the 577 declined fixes and emitted at least one supported change on 38, or 6.6\% of the applied declined set. The combined accepted-or-candidate-covered count is 168/707, or 23.8\%.

This count is candidate coverage, not solved-task coverage. The 38 extractor-positive fixes did not receive a complete fix-level manual precision review. A supported signature change can be incidental to the main issue. The aggregate artifact does not identify the one unapplied item. We therefore report the result as a description of the pipeline rather than a general ceiling on what temporal memory can represent.

The model taxonomy of the 577 declined fixes assigns 6 to value changes, 1 to renames, 109 to structural changes, 405 to logic or control flow and 56 to other. The classifier sees at most 1,500 issue characters and 2,500 patch characters; its parser takes the first category digit in the response and defaults to logic if none is found. An 82-item stratified author audit agrees with 32 labels, with Cohen's $\kappa=.144$ \citep{cohen1960kappa}. The retained scoring record attributes these labels to the study author. The annotator saw the model label as well as truncated issue and patch excerpts, the sample oversampled rare categories, and there was no independent second annotator or adjudication.

These limitations rule out using the taxonomy as a human-validated population estimate. They also make confidence intervals around its model-label proportions an incomplete description of uncertainty. The general lesson is that the accepted atomic slice is limited under the evaluated pipeline; a precise universal percentage would require a stronger annotation and task-success study.

\section{Implications for memory systems}
\subsection{Authority belongs at the write boundary}
An input should distinguish a source observation from an authoritative change when the source can provide that information. The memory can then bind effective time and provenance to a transition instead of inferring world change from arrival order. Sources without this authority can remain evidence candidates or produce an unresolved state.

An authoritative signal cannot simply be an arbitrary string in a retrieved document. The producer or integration must establish its origin and scope. Otherwise a malicious or obsolete source can promote its own content by claiming to be a revert or correction.

\subsection{Current state and retained history serve different consumers}
The active-value view is optimized for questions about current state. The historical view supports reconstruction and explanation. Keeping those views separate does not require deleting history. It requires the consumer to request the appropriate mode and the context builder to honor it.

\subsection{A fallback should state its guarantee}
When a change is not representable as an atomic value, invalidation plus current source is an appropriate weaker contract. It tells the reader that a derived summary may no longer apply. It does not turn a static-analysis approximation into a proof of runtime behavior.

This distinction also clarifies evaluation. Test supersession on declared atomic changes, invalidation on source mutations and answer reasoning on the evidence actually delivered. A single aggregate accuracy score cannot establish all three properties.

\section{Limitations and reproducibility}
The accepted dataset is selected for clean atomic transitions. The re-read adversary deliberately supplies an obsolete statement after the current one. It isolates a meaningful failure mode but does not estimate its prevalence in deployed agents. The guard's revert cost follows from its policy and needs a separate real-history evaluation if a deployment intends to enable it broadly.

The model and scoring pipeline are limited to the retained local configurations. Temperature zero and a fixed seed do not eliminate all runtime nondeterminism. Several supporting measurements are aggregate summaries, so no new paired significance claims are made for them. The 40-case logic experiment is synthetic, and the taxonomy audit is neither independent nor sufficient to validate all labels.

The supplement includes the 130-scenario cache, a 707-ID population inventory, the 577 model labels and 82 annotation pairs, the aggregate contract-extraction checkpoint, relevant log excerpts, study-specific source snapshots and an offline verifier. It also includes a small executable specification of the observation policies and counterexample; that specification illustrates the stated rules and is not the proprietary experimental engine. Exported source snapshots postdate some runs and are reference material, not proof of a fully pinned execution image. A provenance manifest records source hashes and export transformations. Full benchmark patches and individual answer/judge receipts are not included. Reproducing arithmetic and policy examples is distinct from executing the experiments again.

The scenario cache contains a field named \texttt{fixed\_commit} whose historical loader could fall back to environment-setup metadata. It is not a verified patch-commit identifier and must not be used as such. The experiments here ingest rendered statements rather than checking out that field. No claim of validated revert detection from this cache is made.

\section{Conclusion}
The last observation of a value need not be the last change to that value. On 130 accepted software transitions, that distinction turns a strong ordinary temporal-memory result into a severe stale-answer failure when an old statement is re-read. A no-resurrect guard repairs the measured re-read case but cannot resolve a genuine revert from the same observations alone.

The practical requirement is explicit change provenance. When that information is unavailable, the system must expose its policy assumption or uncertainty. Current answers should also remain separate from deliberate historical explanations, and non-atomic changes should trigger freshness handling appropriate to their representation. These boundaries make temporal memory more useful precisely because they specify what it can and cannot establish.

\section*{Ethics and AI assistance}
The source material comes from public software-maintenance benchmarks. The evaluation uses derived statements and synthetic examples, not private user conversations, and does not certify the safety of any software change. The original source projects and benchmark maintainers retain their rights; dataset identities and acquisition references accompany the exports. The single-author taxonomy check is an internal, anchored audit, not independent human validation.

AI assistants were used for manuscript drafting and revision, evaluation-code work and artifact checks under the author's direction. The author is responsible for the claims, citations, code and submitted material. The author has a commercial interest in the memory system studied. Automated correctness judgments, selection by the system's own extractor and disclosure of unavailable records are material limitations, not guarantees of correctness. In deployment, a false authority signal or an over-broad key can suppress valid evidence; an unauthenticated document string must not be treated as trusted change provenance.

\bibliographystyle{plainnat}
\bibliography{references}

\begin{thebibliography}{14}
\providecommand{\natexlab}[1]{#1}
\providecommand{\url}[1]{\texttt{#1}}
\expandafter\ifx\csname urlstyle\endcsname\relax
  \providecommand{\doi}[1]{doi: #1}\else
  \providecommand{\doi}{doi: \begingroup \urlstyle{rm}\Url}\fi

\bibitem[Alchourr{\'o}n et~al.(1985)Alchourr{\'o}n, G{\"a}rdenfors, and
  Makinson]{agm1985}
Carlos~E. Alchourr{\'o}n, Peter G{\"a}rdenfors, and David Makinson.
\newblock On the logic of theory change: Partial meet contraction and revision
  functions.
\newblock \emph{The Journal of Symbolic Logic}, 50\penalty0 (2):\penalty0
  510--530, 1985.
\newblock \doi{10.2307/2274239}.

\bibitem[Cohen(1960)]{cohen1960kappa}
Jacob Cohen.
\newblock A coefficient of agreement for nominal scales.
\newblock \emph{Educational and Psychological Measurement}, 20\penalty0
  (1):\penalty0 37--46, 1960.
\newblock \doi{10.1177/001316446002000104}.

\bibitem[Doyle(1979)]{doyle1979tms}
Jon Doyle.
\newblock A truth maintenance system.
\newblock \emph{Artificial Intelligence}, 12\penalty0 (3):\penalty0 231--272,
  1979.
\newblock \doi{10.1016/0004-3702(79)90008-0}.

\bibitem[Falleri et~al.(2014)Falleri, Morandat, Blanc, Martinez, and
  Monperrus]{falleri2014gumtree}
Jean-R{\'e}my Falleri, Flor{\'e}al Morandat, Xavier Blanc, Matias Martinez, and
  Martin Monperrus.
\newblock Fine-grained and accurate source code differencing.
\newblock In \emph{Proceedings of the 29th ACM/IEEE International Conference on
  Automated Software Engineering}, pages 313--324, 2014.
\newblock \doi{10.1145/2642937.2642982}.

\bibitem[Fluri et~al.(2007)Fluri, W{\"u}rsch, Pinzger, and
  Gall]{fluri2007changedistiller}
Beat Fluri, Michael W{\"u}rsch, Martin Pinzger, and Harald~C. Gall.
\newblock Change distilling: Tree differencing for fine-grained source code
  change extraction.
\newblock \emph{IEEE Transactions on Software Engineering}, 33\penalty0
  (11):\penalty0 725--743, 2007.
\newblock \doi{10.1109/TSE.2007.70731}.

\bibitem[Hui et~al.(2024)Hui, Yang, Cui, Yang, Liu, Zhang, et~al.]{hui2024qwen}
Binyuan Hui, Jian Yang, Zeyu Cui, Jiaxi Yang, Dayiheng Liu, Lei Zhang, et~al.
\newblock Qwen2.5-coder technical report, 2024.
\newblock URL \url{https://arxiv.org/abs/2409.12186}.

\bibitem[Jimenez et~al.(2024)Jimenez, Yang, Wettig, Yao, Pei, Press, and
  Narasimhan]{jimenez2024swebench}
Carlos~E. Jimenez, John Yang, Alexander Wettig, Shunyu Yao, Kexin Pei, Ofir
  Press, and Karthik Narasimhan.
\newblock {SWE-bench}: Can language models resolve real-world {GitHub} issues?
\newblock In \emph{The Twelfth International Conference on Learning
  Representations}, 2024.
\newblock URL \url{https://openreview.net/forum?id=VTF8yNQM66}.

\bibitem[Kowalski and Sergot(1986)]{kowalski1986events}
Robert Kowalski and Marek Sergot.
\newblock A logic-based calculus of events.
\newblock \emph{New Generation Computing}, 4\penalty0 (1):\penalty0 67--95,
  1986.
\newblock \doi{10.1007/BF03037383}.

\bibitem[Lewis et~al.(2020)Lewis, Perez, Piktus, Petroni, Karpukhin, Goyal,
  Kuttler, Lewis, Yih, Rockt{\"a}schel, Riedel, and Kiela]{lewis2020retrieval}
Patrick Lewis, Ethan Perez, Aleksandra Piktus, Fabio Petroni, Vladimir
  Karpukhin, Naman Goyal, Heinrich Kuttler, Mike Lewis, Wen-tau Yih, Tim
  Rockt{\"a}schel, Sebastian Riedel, and Douwe Kiela.
\newblock Retrieval-augmented generation for knowledge-intensive {NLP} tasks.
\newblock In \emph{Advances in Neural Information Processing Systems},
  volume~33, pages 9459--9474, 2020.

\bibitem[Meyer(1992)]{meyer1992contract}
Bertrand Meyer.
\newblock Applying ``design by contract''.
\newblock \emph{Computer}, 25\penalty0 (10):\penalty0 40--51, 1992.
\newblock \doi{10.1109/2.161279}.

\bibitem[Nussbaum et~al.(2024)Nussbaum, Morris, Duderstadt, and
  Mulyar]{nussbaum2024nomic}
Zach Nussbaum, John~X. Morris, Brandon Duderstadt, and Andriy Mulyar.
\newblock Nomic embed: Training a reproducible long context text embedder,
  2024.
\newblock URL \url{https://arxiv.org/abs/2402.01613}.

\bibitem[Snodgrass(1995)]{snodgrass1995tsql2}
Richard~T. Snodgrass, editor.
\newblock \emph{The {TSQL2} Temporal Query Language}.
\newblock Kluwer Academic Publishers, Boston, 1995.
\newblock \doi{10.1007/978-1-4615-2289-8}.

\bibitem[Yadav(2026{\natexlab{a}})]{yadav2026software}
Neeraj Yadav.
\newblock Temporal validity on real software histories: Eliminating stale-fact
  errors in code-assistant memory over github fixes, 2026{\natexlab{a}}.
\newblock URL \url{https://arxiv.org/abs/2608.20685}.

\bibitem[Yadav(2026{\natexlab{b}})]{yadav2026temporal}
Neeraj Yadav.
\newblock Temporal validity in retrieval memory: Eliminating stale-fact errors
  for ai agents over evolving knowledge, 2026{\natexlab{b}}.
\newblock URL \url{https://arxiv.org/abs/2606.26511}.

\end{thebibliography}
\appendix
\section{Retained protocol and measurement contract}\label{app:protocol}
All 130 current-value targets remain $B$ in the ordinary and re-read conditions. The cached sentences were admitted by the earlier atomic-extraction pipeline; this is not a new test of extracting transitions from raw repositories. The guard-enabled experiment reran the temporal condition only. The two retrieval baselines in Table~\ref{tab:reread} are carried forward unchanged from the guard-disabled run, not counted as new replications.

The forced-answer prompt in the retained source snapshot is:
\begin{lstlisting}
Answer the QUESTION with the single most specific value from the CONTEXT.
If the CONTEXT shows MORE THAN ONE value for the thing asked, choose the
one you believe is CURRENT. You MUST give a specific value - never reply
'I don't know', never abstain.

CONTEXT:
<context>

QUESTION: <question>
ANSWER:
\end{lstlisting}
Here the dash is typographically normalized; the exported reference source retains the original prompt bytes. The correctness prompt requests exactly \texttt{CORRECT} or \texttt{INCORRECT} after presenting the question, gold answer and model answer. A response is counted correct if its uppercase text contains \texttt{CORRECT} but not \texttt{INCORRECT}. The literal stale-value rule is independent of this judge verdict. The aggregate harness rounds rates to four decimals before the report formats three; tables preserve the logged values rather than pretending to recover unretained raw verdicts.

The history intervention directly constructs either ``The subject is $B$'' or ``The subject was $A$; it is now $B$.'' The logic intervention directly supplies stale summary, current code, or warning plus both. Neither intervention measures whether a retrieval system would independently select the correct code or history mode.

\section{Policy illustration and counterexample}
Let $c_k$ be the active value and $R_k$ the set of retired values for a resolved key. With the guard disabled, a different incoming value retires $c_k$ and replaces it. With the guard enabled, an incoming value in $R_k$ leaves $c_k$ unchanged unless a trusted, out-of-band change signal authorizes reactivation. Equal active values merely reinforce the existing state.

Starting empty, observations $A,B,A$ therefore end in $A$ without the guard and $B$ with it. The two worlds in Proposition~1 require opposite answers for precisely these observations. A trusted revert event supplies information absent from this sequence and permits $A$ under the extended policy. The supplement tests these rules, multiple keys, repeated observations and the distinction between an untrusted source claim and trusted authority. These are offline specification tests, not new language-model results or an evaluation of a real commit-graph detector.

\section{Supplement coverage}
\begin{center}
\begin{tabularx}{\linewidth}{@{}lX@{}}\toprule
Included evidence & What the offline checks establish\\\midrule
Scenarios and ID inventory & 707 unique source IDs, 130 accepted IDs and 577 declined IDs; no hidden cohort replacement.\\
Taxonomy and audit pairs & Model-label counts, the 82 sampled IDs, 32 agreements and the reported kappa.\\
Coverage checkpoint & 576 applied, 38 candidate-covered and repository/kind sums; not fix-level precision.\\
Aggregate log excerpts & Reported recency, re-read, history and synthetic-context cells; not independent grading.\\
Policy specification & The stated re-read/revert counterexample and authority-sensitive state transitions.\\
Source snapshots & Inspectable construction and scoring code; no claim that all dependencies or historical execution state are present.\\\bottomrule
\end{tabularx}
\end{center}
The verifier makes no network, embedding, reader or judge call. A successful check establishes consistency of the supplied records and manuscript values, not correctness of the original model judgments. No paired significance test is reconstructed from aggregate-only answer logs.
\end{document}